\documentclass[11pt]{article}
\usepackage{a4wide}

\usepackage{array}
\newcolumntype{P}{>{\raggedleft\arraybackslash}p{25pt}}
\newcolumntype{C}{p{10pt}}

\usepackage{amssymb,amsmath,amsthm,mathtools}
\usepackage{algorithm}

\usepackage{algpseudocode} 
\usepackage{pgfplots}

\newtheorem{theorem}{Theorem}

\usepackage{tikz}
\usetikzlibrary{arrows}

\newcommand{\lp}[1]{\ref{lp-#1}}
\newcommand{\nlp}[1]{\ref{nlp-#1}}
\newcommand{\milp}[1]{\ref{milp-#1}}
\newcommand{\eq}[1]{(\ref{eq-#1})}
\newcommand{\fig}[1]{Figure~\ref{fig:#1}}
\newcommand{\algo}[1]{Algorithm~\ref{algo:#1}}
\newcommand{\tab}[1]{Table~\ref{tab-#1}}
\newcommand\R{\mathcal{R}}
\newcommand\I{\mathcal{I}}
\newcommand\J{\mathcal{J}}
\newcommand\K{\mathcal{K}}
\newcommand\C{\mathcal{C}}
\newcommand\N{\mathcal{N}}

\newcommand\B{\mathcal{B}}
\newcommand\V{\mathcal{V}}
\newcommand\Penal{\mathcal{P}}
\newcommand\A{\mathcal{A}}

\newcommand{\nikos}[1]{{#1}}

\title{Fast Reconstruction of Compact Context-Specific Metabolic Network Models}

\author{Nikos Vlassis$^*$ \qquad
Maria Pires Pacheco$^\dag$ \qquad
Thomas Sauter$^\dag$\\
~\\
$^*$Luxembourg Centre for Systems Biomedicine, University of Luxembourg \\
$^\dag$Life Sciences Research Unit, University of Luxembourg}

\date{\today} 

\begin{document}
\maketitle

\begin{abstract}
Systemic approaches to the study of a biological cell or tissue rely increasingly on the use of context-specific metabolic network models. The reconstruction of such a model from high-throughput data can routinely involve large numbers of tests under different conditions and extensive parameter tuning, which calls for fast algorithms. 
We present {\rm\textsc{fastcore}}, a generic algorithm for reconstructing context-specific metabolic network models from global genome-wide metabolic network models such as Recon\,X. {\rm\textsc{fastcore}} takes as input a core set of reactions that are known to be active in the context of interest (e.g., cell or tissue), and it searches for a flux consistent subnetwork of the global network that contains all reactions from the core set and a minimal set of additional reactions. 
Our key observation is that a minimal consistent reconstruction can be defined via a set of sparse modes of the global network, and {\rm\textsc{fastcore}} iteratively computes such a set via a series of linear programs. Experiments on liver data demonstrate speedups of several orders of magnitude, and significantly more compact reconstructions, over a rival method. Given its simplicity and its excellent performance, {\rm\textsc{fastcore}} can form the backbone of many future metabolic network reconstruction algorithms.
\end{abstract}

\section{Introduction}

Cell metabolism is known to play a key role in the pathogenesis of various diseases \cite{Debardinis12} such as Parkinson's disease \cite{Pourfar13} and cancer \cite{Hiller13}. 
The study of human metabolism has been greatly advanced by the development of computational models of metabolism, such as Recon\,1 \cite{Duarte07}, the Edinburgh human metabolic network  \cite{Hao10}, and Recon\,2 \cite{Thiele13}. These are genome-scale metabolic network models that have been reconstructed by combining various sources of `omics' and literature data, and they involve a large set of biochemical reactions that can be active in different contexts, e.g., different cell types or tissues \cite{Thiele10}.

To maximize the predictive power of a metabolic model when conditioning on a specific context, for instance the energy metabolism of a neuron or the metabolism of liver, recent efforts go into the development of {\em context-specific} metabolic models \cite{Becker08,Christian09,Jerby10,Chang10,Lewis10,Agren12}. These are network models that are derived from global models like Recon\,1, but they only contain a subset of reactions, namely, those reactions that are active in the given context. Such context-specific metabolic models are  known to exhibit superior explanatory and predictive power than their global counterparts \cite{Jerby10,Folger11,Bordbar12}.

Most algorithms for context-specific metabolic network reconstruction \nikos{(see Section 2.5 for a short overview)} first identify a relevant subset of reactions according to some `omics' information (typically expression data and bibliomics), and then search for a subnetwork of the global network that satisfies some mathematical requirements and contains all (or most of) these reactions \cite{Becker08,Shlomi08,Jerby10,Price10,Jensen11,Agren12}. 
The mathematical requirements are typically imposed via flux balance analysis, which characterizes the steady-state distribution of fluxes in a metabolic network via linear constraints that are derived from the stoichiometry of the network and physical conservation laws \cite{Schuster94,Stephanopoulos98,Price04,Gagneur04,Fleming12}. 
The search problem may target the optimization of a specific functionality of the model (e.g., biomass production) or some other objective~\cite{Blazier12}, and it may involve repeated tests under different conditions and parameter tuning \cite{Becker08,Folger11,Orth11,Wang12}. The latter calls for fast algorithms.

We present {\rm\textsc{fastcore}}, a generic algorithm for context-specific metabolic network reconstruction. {\rm\textsc{fastcore}} takes as input a core set of reactions that are supported by strong evidence to be active in the context of interest. Then it searches for a {\em flux consistent} subnetwork of the global network that contains all reactions from the core set and a minimal set of additional reactions. Flux consistency implies that each reaction of the network is active (i.e., has nonzero flux) in at least one feasible flux distribution \cite{Schuster94,Acuna09}. An attractive feature of {\rm\textsc{fastcore}} is its generality: As it only relies on a preselected set of reactions and a simple mathematical objective (flux consistency), it can be applied in different contexts and it allows the integration of different pieces of evidence (`multi-omics') into a single model. 

Computing a minimal consistent reconstruction from a subset of reactions of a global network is, however, an NP-hard problem \cite{Acuna09}, and hence some approximation is in order. 
Our key observation is that a minimal consistent reconstruction can be defined via a set of {\em sparse} modes of the global network, and {\rm\textsc{fastcore}} is designed to compute a minimal such set.
Every iteration of the algorithm computes a new sparse mode via two linear programs that aim at maximizing the support of the mode inside the core set while minimizing that quantity outside the core set. {\rm\textsc{fastcore}}'s search strategy is in marked contrast to related approaches, in which the search for a minimal consistent reconstruction involves, for instance, incremental network pruning \cite{Jerby10}. 
{\rm\textsc{fastcore}} is simple, devoid of free parameters, and its performance is excellent in practice: As we demonstrate on experiments with liver data, {\rm\textsc{fastcore}} is several orders of magnitude faster, and produces much more compact reconstructions, than the main competing algorithm MBA \cite{Jerby10}.

\section{Methods}

\subsection{Background}

A metabolic network of $m$ metabolites and $n$ reactions is represented by an $m \times n$  {\em stoichiometric} matrix $S$, where each entry $S_{ij}$ contains the stoichiometric coefficient of metabolite $i$ in reaction $j$. 
A {\em flux} vector $v\in\mathbb{R}^n$ is a tuple of reaction rates, $v = (v_1,\ldots,v_n)$,  where $v_i$ is the rate of reaction $i$ in the network. 
Reactions are grouped into {\em reversible} ones ($\R$) and {\em irreversible} ones ($\I$). For a reaction $i \in \I$ \nikos{it holds that} $v_i \geq 0$; this and other imposed flux bounds, e.g., lower and upper bounds per reaction, are collectively denoted by $\B$ (which defines a convex set).
A flux vector is called {\em feasible} or a {\em mode} if it satisfies a set of steady-state mass-balance constraints that can be compactly expressed as:
\begin{equation}
  Sv = 0, \quad v \in \B \, .
  \label{eq-mb}
\end{equation}
An {\em elementary} mode is a feasible flux vector $v \neq 0$ with minimal support, that is, there is no other feasible flux vector $w \neq 0$ with $supp(w) \subset supp(v)$, where \nikos{$supp(v) = \big\{j \in \{1,2,\ldots, n \}: v_j \neq 0 \big\}$} is the support (i.e., the set of nonzero entries) of $v$ \cite{Schuster94,Gagneur04}.
\nikos{A reaction $i$ is called {\em blocked} if it cannot be active under any mode, that is, there exists no mode $v\in\mathbb{R}^n$  such that $v_i \neq 0$ (in practice $|v_i| \ge \varepsilon$, for some small positive threshold $\varepsilon$).
A metabolic network model that contains no blocked reactions is called {\em (flux) consistent} \cite{Schuster94,Acuna09}.
}

\subsection{Network consistency testing}

Given a metabolic network model with stoichiometric matrix $S$, a problem of interest is to test whether the network is consistent or not. Additionally, if the network is inconsistent, it would be desirable to have a method that detects \nikos{all blocked reactions}.

\begin{figure}
\centering
\begin{tikzpicture}[scale=0.69,line join=bevel]
\large%
\node (A) at (86bp,56bp) [draw,circle] {A};
  \node (C) at (165bp,10bp) [draw,circle] {C};
  \node (B) at (165bp,102bp) [draw,circle] {B};
  \node (D) at (245bp,56bp) [draw,circle] {D};
  \node (node1) at (9bp,56bp) [draw,circle,densely dotted] {~};
  \node (node2) at (323bp,56bp) [draw,circle,densely dotted] {~};
  \draw [->,thick,] (C) ..controls (186.42bp,22.041bp) and (210.36bp,36.163bp)  .. (D);
  \definecolor{strokecol}{rgb}{0.0,0.0,0.0};
  \pgfsetstrokecolor{strokecol}
  \draw (211bp,22bp) node {$v_5$};
  \draw [->,thick] (A) ..controls (114.95bp,56bp) and (154bp,56bp)  .. (154bp,56bp) .. controls (154bp,56bp) and (197.33bp,56bp)  .. (D);
  \draw (165bp,65bp) node {$v_3$};
  \draw [->,thick] (D) ..controls (267.59bp,56bp) and (288.28bp,56bp)  .. (node2);
  \draw (285bp,65bp) node {$v_6$};
  \draw [<->,thick] (A) ..controls (107.17bp,68.049bp) and (130.85bp,82.194bp)  .. (B);
  \draw (117bp,88bp) node {$v_2$};
  \draw [->,thick] (node1) ..controls (29.338bp,56bp) and (50.316bp,56bp)  .. (A);
  \draw (47bp,65bp) node {$v_1 $};
  \node (nodetmp) at (66bp,47bp)  {${\scriptscriptstyle 2}$};
  \draw [->,thick] (A) ..controls (107.17bp,43.951bp) and (130.85bp,29.806bp)  .. (C);
  \draw (117bp,24bp) node {$v_4$};
\end{tikzpicture}
\caption{A metabolic network with one \nikos{blocked} reaction (A$\leftrightarrow$B).
\nikos{Note that A appears with stoichiometric coefficient 2 in the boundary reaction $\rightarrow$2A.}
\label{fig:example}}
\end{figure}
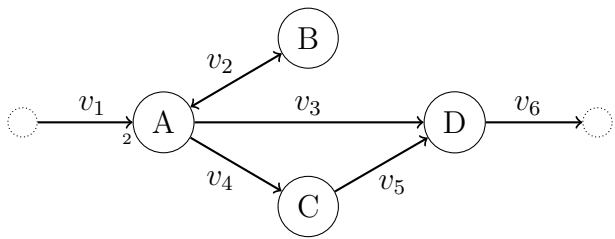

It has been suggested that network consistency can be detected by a single linear program (LP) \cite{Acuna09}. The idea is to first convert each reversible reaction into two irreversible reactions (and define a reversible flux as the difference of two irreversible fluxes), and then test if the minimum feasible flux on the new set $\J$ of irreversible-only reactions is strictly positive (in practice, at least~$\varepsilon$). This is equivalent to testing if the following LP is feasible:
\begin{equation}   	  
  \begin{aligned} 
  \max_{v,z}  	&& z 	&&& \\
  \mbox{s.t.} 	&& z 	& \ge \varepsilon	& \quad & z \in \mathbb{R} \\
  				&& v_i 	& \ge z	 			& \quad & \forall i \in \J  \\
				&& Sv 	& = 0 				& \quad & v \in \B \, .
  \end{aligned}
  \tag{LP-\theequation} \label{lp-acuna} \stepcounter{equation}
\end{equation}
This test of consistency, however, can produce spurious solutions. In \fig{example} we show a toy metabolic network comprising four metabolites (A,B,C,D) and six reactions annotated with corresponding fluxes $v_1,\ldots, v_6$. \nikos{Fluxes are bounded as $0 \leq v_i \leq 3$ for $i\neq 2$, and $|v_2|\leq 3$}. All stoichiometric coefficients are equal to one, except for the {reaction $\rightarrow$2A.} The only reversible reaction is A$\leftrightarrow$B, which is a dead-end reaction and therefore \nikos{blocked}, whereas all other reactions are irreversible and \nikos{unblocked}. After converting A$\leftrightarrow$B to a pair of irreversible reactions, \lp{acuna} achieves optimal value $z^*=1.5$, which implies (wrongly) that the network is consistent. The test here fails because the two irreversible copies of A$\leftrightarrow$B have equal flux at the solution, thereby nullifying the actual net flux of A$\leftrightarrow$B.

A straightforward solution to the problem would involve iterating through all reactions, computing the maximum and minimum feasible flux of each reaction via an LP that satisfies the constraints in \eq{mb}. 
\nikos{Reactions with minimum and maximum flux zero would then be blocked.}
This is the idea behind the FVA \nikos{(Flux Variability Analysis)} algorithm and the {\em reduceModel} function of the COBRA toolbox \cite{Mahadevan03,COBRA2}. However, iterating through all reactions can be inefficient. A faster variant is fastFVA \cite{Gudmundsson10}, which achieves acceleration over FVA via LP warm-starts. Another fast algorithm is CMC \nikos{(CheckModelConsistency)} \cite{Jerby10}, which involves a series of LPs, where each LP maximizes the sum of fluxes over a subset $\J$ of reactions:
\begin{equation}
  \begin{aligned} 
    \max_{v} 		&& \sum_{j \in \J} v_j  \hspace{-15pt} \\
    \mbox{s.t.}	&& Sv 	& = 0		& \quad & v \in \B \, .
  \end{aligned}
  \tag{LP-\theequation} \label{lp-jerby} \stepcounter{equation}
\end{equation}
The set $\J$ is initialized by $\J = \R \cup \I$ (all reactions in the network), and it is updated after each run of \lp{jerby} so that it contains the reactions whose consistency has not been established yet. When $\J$ cannot be reduced any further, we can reverse the signs of the columns of~$S$ corresponding to the reversible reactions in $\J$ and resume the iterations. Eventually, all remaining reactions may have to be tested one by one for consistency, as in FVA. Such an iterative scheme is complete, in the sense that it will always report consistency if the network is consistent, and if not, it will reveal the set of \nikos{blocked} reactions. However, as we will clarify in the next section, \lp{jerby} is not optimizing the `correct' function, which may result in unnecessarily many iterations. For example, when applied to the network of \fig{example}, \lp{jerby} will pick up the elementary mode that corresponds to the pathway A$\rightarrow$C$\rightarrow$D (because this pathway achieves maximum sum of fluxes {$v_1 + v_4 + v_5 + v_6 = 1.5 + 3 + 3 + 3$}), and it will set $v_3=0$. To establish the consistency of the reaction A$\rightarrow$D, an additional run of \lp{jerby} would be needed, where the set $\J$ would only involve  the reactions A$\leftrightarrow$B and A$\rightarrow$D. Hence, an iterative algorithm like CMC that relies on \lp{jerby} would need two iterations to detect the consistent part of this network. However, one LP suffices to detect \nikos{the consistent subnetwork} in this example, as we explain in the next section. {In more general problems involving larger and more realistic networks, CMC may involve unnecessarily many iterations, as we demonstrate in the experiments.}

\subsection{Fast \nikos{consistency testing}}

In most problems of interest there will be no single mode that renders the whole network consistent, and an iterative algorithm like the one described in the previous section must be used. For performance reasons it would therefore be desirable to be able to establish the consistency of as many reactions as possible in each iteration of the algorithm. 

Since consistency implies nonzero fluxes, it is sufficient to optimize a function that just `pushes' all fluxes away from zero. Formally, this amounts to searching for modes $v$ whose {\em cardinality}---denoted by $card(v)$ and defined as \nikos{$card(v)= \# supp(v)$}, i.e., the number of nonzero entries of $v$---is as large as possible. Directly maximizing $card(v)$ is, however, not straightforward, \nikos{for the following reasons}: First, the $card$ function is quasiconcave only for $v\in\mathbb{R}^n_+$ (the nonnegative orthant), and it is nonconvex for general $v\in\mathbb{R}^n$ \cite{Boyd04}. Second, even if we restrict attention to nonnegative fluxes in each iteration (which we can do without loss of generality by flipping the signs of the corresponding columns of $S$), it is not obvious how to efficiently maximize the quasiconcave $card(v)$. Third, in practice consistency implies fluxes that are $\varepsilon$-distant from zero, in which case some adaptation of the $card$ function is in order.

\nikos{Here we propose an approach to approximately maximize $card(v)$ over a nonnegative flux subspace indexed by a set of reactions $\J$. First note that the cardinality function can be expressed as 
\begin{equation}
  card(v) = \sum_{i \in \J} \theta(v_i) \, ,
  \label{eq-card}
\end{equation}
}
where $\theta: \mathbb{R} \rightarrow \{0,1\}$ is a step function:
\begin{equation}
  \theta(v_i) = \left\{ 
  \begin{array}{ll}
  0 & \mbox{if} \ v_i=0 \\
  1 & \mbox{if} \ v_i>0 \, .
  \end{array} 
  \right. 
\end{equation}
The key idea is to approximate the function $\theta$ by a concave function that is the minimum of a linear function and a constant function:
\begin{equation}
 \theta(v_i) \approx \min \{ \frac{v_i}{\varepsilon}, 1 \}  \, ,
 \label{eq-theta}
\end{equation}
where $\varepsilon$ is the flux threshold. 
\nikos{The problem of approximately maximizing $card(v)$ can then be cast as an LP: We introduce an auxiliary variable $z_i \in \mathbb{R}_+$ for each flux variable $v_i$, for $i \in \J$, and take epigraphs \cite{Boyd04}, in which case maximizing $card(v) = \sum_{i \in \J} \theta(v_i)$ can be expressed as
\begin{equation*}
  \begin{aligned} 
    \max_{v,z} 	&& \sum_{i \in \J} z_i \hspace{-15pt} \\
    \mbox{s.t.} 	&& z_i 	& \leq \theta(v_i) 			& \quad & \forall i \in \J, \ z_i \in\mathbb{R}_+ \\
    				&& v_i 	& \geq 0 		& \quad & \forall i \in \J \\
                && Sv 	& = 0  						& \quad & v \in \B \, .
  \end{aligned}
\end{equation*}
Using \eq{theta} and assuming constant $\varepsilon$, this simplifies to
}
\begin{equation}
  \begin{aligned} 
    \max_{v,z} 	&& \sum_{i \in \J} z_i \hspace{-15pt} \\
    \mbox{s.t.} 	&& z_i 	& \in [0,\varepsilon] 			& \quad & \forall i \in \J, \ z_i \in\mathbb{R}_+ \\
				&& v_i 	& \geq z_i 		& \quad & \forall i \in \J \\
                && Sv 	& = 0  						& \quad & v \in \B \, .
  \end{aligned}
  \tag{LP-\theequation} \label{lp-cardinality} \stepcounter{equation}
\end{equation}
Note that \lp{cardinality} tries to maximize the number of feasible fluxes in $\J$ whose value is at least $\varepsilon$ {(contrast this with \lp{acuna})}.

Returning to the network of \fig{example}, if $\J$ comprises all network reactions, then note that the flux vector $[v_1,v_2,v_3,v_4,v_5,v_6] = [\varepsilon,0,\varepsilon,\varepsilon,\varepsilon,2\varepsilon]$ is an optimal solution of \lp{cardinality}. Hence, a single run of the latter can \nikos{detect all unblocked reactions} of that network. More generally, a single run of \lp{cardinality} on an arbitrary subset $\J$ of a given network will typically \nikos{detect all unblocked {\em irreversible}} reactions of $\J$.
The intuition is that \lp{cardinality} prefers flux `splitting' over flux `concentrating' in order to maximize the number of participating reactions in the solution, which, in the case of irreversible reactions, corresponds to flux cardinality maximization. 

By construction, the above approximation of the cardinality function applies only to nonnegative fluxes. In order to deal with reversible reactions that can also take negative fluxes, we can embed \lp{cardinality} in an iterative algorithm (as in the previous section), in which reversible reactions are first considered for positive flux via \lp{cardinality}, and then they are considered for negative flux. The latter is possible by flipping the signs of the columns of the stoichiometric matrix that correspond to the reversible reactions under testing, in which case the fluxes of the transformed model are again all nonnegative, and the above approximation of the cardinality function can be used. 
This gives rise to an algorithm for detecting the consistent part of a network that we call {\rm \textsc{fastcc}} (for fast consistency check). Since {\rm \textsc{fastcc}} is just a variant of {\rm\textsc{fastcore}}, we defer its detailed description until the next section.

\nikos{Independently to this work, a similar approach to network consistency testing was recently proposed, called OnePrune \cite{Dreyfuss13}. OnePrune first converts each reversible reaction into two irreversible reactions, forming an augmented set $\J$ of irreversible-only reactions (as in \lp{acuna} above), and then it employs an LP that coincides with \lp{cardinality} for the above choice of $\J$ and $\varepsilon=1$. However, such an approach is prone to the same drawback as \lp{acuna}, namely, that the two irreversible copies of a blocked reaction can carry equal positive flux at the solution of \lp{cardinality} due to the presence of cycles introduced by the transformation. The authors acknowledge this problem but they do not fully resolve it. In our case, we avoid this problem by working with the original reactions and a series of LPs with appropriate sign flips of the stoichiometric matrix, thereby guaranteeing the completeness of the  algorithm.}

\subsection{Context-specific network reconstruction}

The reconstruction problem involves computing a minimal consistent network from a \nikos{global network and a `core' set of reactions} that are known to be active in a given context. Formally, given \nikos{(i) a {\em consistent} global network $\{ \N, S_\N  \}$ with reaction set $\N=\{1,2,\ldots,n\}$ and stoichiometric matrix $S_\N$, and (ii) a set $\C \subset \N$, the problem is to find the smallest set $\A \subseteq \N$ such that $\C\subseteq\A$ and the subnetwork $\{ \A, S_\A  \}$ induced by the reaction set $\A$ is consistent. (By $S_{\!\A}$ we denote the submatrix of $S_\N$ that contains only the columns indexed by $\A$.)} This problem is known to be NP-complete~\cite{Acuna09}, suggesting that a practical solution should entail some approximation. (We note that Acu{\~n}a et al.~\cite{Acuna09} prove NP-completeness of this problem by noting that a special case involves $\C$ being the empty set, in which case the problem comes down to finding the smallest elementary mode of \nikos{the global network}, which, as the authors show, is NP-complete. However, this leaves open the case of a nonempty core set $\C$, since a solution to the minimal reconstruction problem need not constitute an elementary mode. We conjecture that the problem remains NP-hard when $\C$ is nonempty, but we are not pursuing this question here.)

Our approach hinges on the observation that a \nikos{consistent induced subnetwork of the global network can be defined via a set of modes of the latter}: 
\begin{theorem}
Let $\V$ be a set of modes of \nikos{the global network $\{ \N, S_\N  \}$}, and let $\A = \cup_{v \in \V} \ supp(v)$ be the union of the supports of these modes. \nikos{The induced subnetwork $\{ \A, S_\A  \}$ is consistent.}
\end{theorem} 
\begin{proof}  
For each $v \in \V$, let $v_{\!\A}$ be the `truncated' $v$ after dropping all dimensions not indexed by $\A$. Clearly, $S_{\!\A} \hspace{1pt} v_{\!\A} =0$, therefore each $v_{\!\A}$ is a mode in the reduced model $\{\A, S_{\!\A}\}$. 
By construction of $\A$, each reaction in $\A$ is in the support of some $v \in \V$, and hence also in the support of some mode $v_{\!\A}$ of the reduced model. 
\end{proof}
This simple result allows one to cast the reconstruction problem as a search problem over sets of modes of the global network: 
\begin{equation}
  \begin{aligned} 
    \min_{\V} 	&& card( \A ) \hspace{-25pt} \\
    \mbox{s.t.} 	&& \A & = \bigcup_{v \in \V} \ supp(v) \\
				&& \C & \subseteq \A \\ 
				&& \nikos{\forall} & \nikos{v \in\V : \ \  S_{\!\N}^{} \hspace{1pt} v  = 0, \ v \in \B  }\,.
  \end{aligned}
  \tag{NLP-\theequation} \label{nlp-1} \stepcounter{equation}
\end{equation}
\nikos{Note that this optimization problem involves searching for a set $\V$ of modes of $\{ \N, S_\N  \}$, such that the union of the support of these modes (the set $\A$) is a minimal-cardinality set that contains the core set $\C$. In order to practically make use of this theorem, one has to define a search strategy over modes. Next we discuss two possibilities. The first gives rise to an exact algorithm, but this algorithm does not scale to large networks. The second is a scalable greedy approach that gives rise to {\rm\textsc{fastcore}}.}

\subsubsection*{\nikos{Exact reconstruction via mixed integer linear programming}}

\nikos{Note that, without loss of generality, in \nlp{1} we can restrict the search for $\V$ over all {\em elementary modes} of the global network, since the union of their supports covers the whole set $\N$. As we show next, if all elementary modes are available, \nlp{1} can be cast as a mixed integer linear program (MILP) and solved exactly. This MILP is defined as follows. Let $r$ be the number of elementary modes, and $\{m_1,\ldots,m_r\}$ be a set of length-$n$ binary vectors, where each vector $m_j$ captures the support of elementary mode $j$ (so, its $i$th entry is 1 if reaction $i$ is included in elementary mode $j$, and 0 otherwise). Also, let $c=(c_1,\ldots,c_n)$ be a length-$n$ binary vector with $c_i=1$ if reaction $i$ is included in the core set $\C$, and $c_i=0$ otherwise. The decision variables of the MILP are a length-$n$ binary vector $x=(x_1,\ldots,x_n)$ and a length-$r$ real vector $y=(y_1,\ldots,y_r)$. At an optimal solution of the MILP, the set $\A$ is defined as $\A = \{ i \in \N: \ x_i^* = 1 \}$. 
\begin{theorem}
When all elementary modes are available, the following \milp{1} solves \nlp{1} exactly.
\begin{equation}
  \begin{aligned} 
    \min_{x,y} 	&& \sum_i x_i \hspace{-18pt} \\
    \mbox{s.t.} && x & \geq \frac{1}{r} \sum_j m_j y_j \\ 
				&& c & \leq \sum_j m_j y_j \\
				&& y & \in [0,1] \\
				&& x & \in\{0,1\} \,.
  \end{aligned}
  \tag{MILP-\theequation} \label{milp-1} \stepcounter{equation}
\end{equation}
\end{theorem} 
\begin{proof}  
By definition, $x_i^*=1$ implies that reaction $i$ will be included in the reconstruction $\A$, hence the objective minimizes the cardinality of $\A$. The sum $\sum_j m_j y_j^*$ is a vector whose support is the union of the supports of all selected elementary modes at the solution, where an elementary mode $j$ is selected when $y_j^*>0$. The first constraint $x \geq \frac{1}{r} \sum_j m_j y_j$ therefore imposes that the set $\A$ must contain the union of the supports of the selected elementary modes at the solution. (The factor $\frac{1}{r}$ ensures that $\frac{1}{r} \sum_j m_j y_j \leq 1$). Since superfluous reactions are removed by the minimization of $\sum_i x_i$ in the objective, the above implies that $\A$ is precisely the union of the supports of the selected elementary modes at the solution. The second constraint $c \leq \sum_j m_j y_j$ imposes that the core set must be included in the union of the supports of the selected elementary modes at the solution, and hence the core set must be included in $\A$. Therefore, all constraints of \nlp{1} are satisfied at the optimal solution of \milp{1}, and since the two programs minimize the same objective, an optimal solution of \milp{1} must be an optimal solution of \nlp{1}.
\end{proof}
Note, however, that \milp{1} does not scale to large networks, for the following reasons: First, it requires computing all elementary modes of the global network, which can be a very large number \cite{Gagneur04}. Second, the binary decision variables $x_i$ index all reactions of the global network, and therefore \milp{1} needs to search over a binary hypercube of dimension $n$, which can be prohibitive for large $n$. Nonetheless, it is reassuring to know that an exact solution to the context-specific network reconstruction problem is possible, albeit with high complexity. Next we describe {\rm\textsc{fastcore}}, an approximate greedy algorithm that scales much better to large networks, and we compare it to \milp{1} in the Results section.}

\subsubsection*{\nikos{Greedy approximation and the {\rm\textsc{fastcore}} algorithm}}

\nikos{An alternative search strategy for computing $\V$ in \nlp{1} is a greedy approach,  reminiscent of greedy heuristics for the related {\em set covering problem} \cite{Chvatal79}. This is the idea behind the proposed {\rm\textsc{fastcore}} algorithm: We build up the set $\V$ in a greedy fashion, by computing in each iteration a new mode of the global network.} Further, as a means to approximately minimize $card(\A)$, each added mode is constrained to have {\em sparse} support outside $\C$. This is implemented via $L_1$-norm minimization, which is a standard approach to computing sparse solutions to (convex) optimization problems \cite{Boyd04,Julius08}. 
 
\algrenewcommand{\algorithmicrequire}{\textbf{Input:}}
\algrenewcommand{\algorithmicensure}{\textbf{Output:}}
\algrenewcommand{\algorithmicforall}{\textbf{for each}}
\newcommand\ForEach{\ForAll}
    
\begin{algorithm}[t]
\small
\caption{{The {\rm\textsc{fastcore}} algorithm} \label{algo:algo}}
    \begin{algorithmic}[1]
        \Require \nikos{A consistent metabolic network model $\{ \N, S_\N  \}$ and a reaction set $\C \subset \N$.}
        \Ensure \nikos{A consistent induced subnetwork $\{ \A, S_\A  \}$ of $\{ \N, S_\N  \}$ such that $\C \subseteq \A$.}      
        \Statex \vspace*{-6pt}
        \Function{{\rm\textsc{fastcore}}}{$\,\N, \C\,$}
        \State $\J \gets \C\cap\I$, \, $\Penal \gets \N \setminus \C$ \label{step-init}
        \State $flipped \gets False$, $singleton \gets False$
		\State $\A \gets \Call{{\rm \textsc{FindSparseMode}}}{\,\J, \Penal, singleton\,}$
        \State $\J \gets \C \setminus \A$
        \While{$\J \neq \emptyset$\,}  \label{step-loop}
            \State $\Penal \gets \Penal \setminus \A$  \label{step-Pupdate}
            \State $\A \gets \A \cup \Call{{\rm \textsc{FindSparseMode}}}{\,\J, \Penal, singleton\,}$
            \If{$\J \cap \A \neq \emptyset$}
                \State $\J \gets \J \setminus \A$, \ $flipped \gets False$  \label{step-Jupdate}
            \Else
                \If{$flipped$}
                    \State $flipped \gets False$, \ $singleton \gets True$
                \Else
                    \State $flipped \gets True$
            			\If{$singleton$}
                			\State $\tilde\J \gets \J(1)$ \quad (the first element of $\J$)
            			\Else
                			\State $\tilde\J \gets \J$
            			\EndIf                    
                    \ForAll{$i \in \tilde\J \setminus \I$ }
                    \State flip the sign of the $i$'th column of $S_\N$ and \label{step-flipsigns}
                    \State swap the upper and lower bounds of $v_i$
                		\EndFor
                \EndIf
            \EndIf
        \EndWhile
        \State \Return $\A$
                \EndFunction
    \end{algorithmic}
\end{algorithm}
\begin{algorithm}[h!]
\small
\caption{{The \textsc{FindSparseMode} function} 
\label{algo:findsparsemode}}
    \begin{algorithmic}[1]   
        \Require A set $\J \subseteq \C$, a penalty set $\Penal \subseteq \N\setminus\C$, and the $singleton$ flag.
        \Ensure The support of a mode that is dense in $\J$ and sparse in $\Penal$.
        \Statex \vspace*{-6pt}
        \Function{{\rm \textsc{FindSparseMode}}}{$\,\J, \Penal, singleton\,$}
            \If{$\J=\emptyset$}
                \State \Return $\emptyset$
            \EndIf   
            \If{$singleton$}
				\State $v^* \gets$ \lp{cardinality} on set $\J(1)$            
			\Else
				\State $v^* \gets$ \lp{cardinality} on set $\J$            
			\EndIf                     
            \State $\K \gets \{i \in \J : v_i^* \ge \varepsilon \}$ \label{step-K}
            \If{$\K=\emptyset$}
                \State \Return $\emptyset$
            \EndIf
            \State $v^* \gets$ \lp{l1norm} on sets $\K, \Penal$  \label{step-lp9}
            \State \Return $\{i \in \N : |v_i^*| \ge \varepsilon \}$
        \EndFunction
    \end{algorithmic}
\end{algorithm}

The overall {\rm\textsc{fastcore}} algorithm is shown in \algo{algo}. The algorithm maintains a set $\J\subseteq\C$ that is initialized with the irreversible reactions in $\C$, and a `penalty' set $\Penal = (\N\setminus\C)\setminus\A$ that contains all reactions outside $\C$ that have not been added yet to the set $\A$. Each iteration adds to the set $\A$ the support of a mode that is dense in $\J$ (i.e., contains as many nonzero fluxes in $\J$ as possible) and sparse in $\Penal$ (i.e., contains as many zero fluxes in $\Penal$ as possible), computed by the function \textsc{FindSparseMode} (\algo{findsparsemode}). This function first applies an \lp{cardinality} to compute an active subset $\K$ of $\J$, and then it applies the following $L_1$-norm minimization LP constrained by the set $\K$:
\begin{equation}
  \begin{aligned} 
    \min_{v,z} 	&& \sum_{i \in \Penal}  z_i  \hspace{-12pt} \\
    \mbox{s.t.} && v_i 	& \in [-z_i, z_i] 			& \quad & \forall i \in \Penal, \ z_i\in\mathbb{R}_+ \\
                	&& v_i 	& \geq \varepsilon 		& \quad & \forall i \in \K  \\
    				&& \nikos{S_\N^{}} \,v 	& = 0  						& \quad & v \in \B \, .
  \end{aligned}
  \tag{LP-\theequation} \label{lp-l1norm} \stepcounter{equation}
\end{equation}
The \lp{l1norm} minimizes $\sum_{i \in \Penal} \vert v_i \vert$, the $L_1$ norm of fluxes in the penalty set $\Penal$ (expressed via epigraphs), subject to a minimum flux constraint on the set $\K$. However, some care is needed to preempt false negative solutions arising from the minimization of $L_1$ norm in \lp{l1norm}. For example, suppose in the network of \fig{example} that \nikos{the global network} comprises all reactions except A$\leftrightarrow$B, and $\C=\J=\K=\{6\}$ and $\Penal=\{1,3,4,5\}$. In this case, \lp{l1norm} could settle to a solution 
$[v_1,v_3,v_4,v_5,v_6] = [\frac{\varepsilon}{2},\varepsilon,0,0,\varepsilon]$. The flux $v_1$, being below $\varepsilon$, would be treated as zero by \textsc{FindSparseMode}, in which case the reaction $\rightarrow$2A would be erroneously excluded from the reconstruction. A simple way to avoid this is to use a scaled version of $\varepsilon$ (we used $10^5 \varepsilon$) in the second constraint of \lp{l1norm}, with an equal scaling of all flux bounds in~$\B$.

The {\rm\textsc{fastcore}} algorithm first goes through the $\I\cap\C$ reactions (step \ref{step-init}), and then through the $\R\cap\C$ ones (and eventually through each individual reversible reaction in the core set; when $singleton = True$). The $flipped$ variable ensures that a reversible reaction is tested in both the forward and negative direction.
The algorithm terminates when all reactions in $\C$ have been added to $\A$, which is guaranteed since in the main loop the set $\J$ never expands (step \ref{step-Jupdate}) and the global network is consistent.
Note that {\rm\textsc{fastcore}} has no free parameters besides the flux threshold~$\varepsilon$.

The {\rm \textsc{fastcc}} algorithm for detecting the consistent part of an input network (see previous section) can be viewed as a variant of {\rm\textsc{fastcore}}$(\N,\N)$ in which the steps \ref{step-K}--\ref{step-lp9} of \textsc{FindSparseMode} are omitted (and there is no $\Penal$ set). 
It is easy to verify that {\rm \textsc{fastcc}} is complete, in the sense that it will always report consistency if the network is consistent, and if not, it will reveal the set of \nikos{blocked} reactions.

\subsection{\nikos{Related work}}

\begin{table*}[!ht]
\caption{\nikos{Summary of the main characteristics of GIMME \cite{Becker08}, MBA \cite{Jerby10}, iMAT \cite{Zur10}, mCADRE \cite{Wang12}, INIT \cite{Agren12}, and {\rm\textsc{fastcore}} (this paper) reconstruction algorithms.}  \label{tab-table3}}
\begin{center}
\nikos{
 \begin{tabular}{l p{1.5cm} p{1.5cm} p{1.5cm} p{1.5cm} p{1.5cm} p{1.5cm} }  
~&{GIMME}& {MBA}& {iMAT}&{mCADRE}& {INIT} & {\rm\textsc{fastcore}} \\
\hline
Optimization  & LP 	& MILP 	& MILP 	& MILP    & MILP  & LP \\ 
Computational cost     & low 	&  high 	&  high 	&    high &  high & low \\ 
Function required & yes	&   no 	& no & 	yes 	& yes & no \\ 
Omics required & yes	&   optional 	& yes & yes & yes & no	\\ 
Code available & yes	&   yes 	& yes&   yes & no & yes \\ 
\hline
\end{tabular}
}
\end{center}
\end{table*}

\nikos{Several algorithms have been published in the last years for extracting condition-specific models from generic genome-wide models like Recon\,1. Among them, mCADRE \cite{Wang12}, INIT \cite{Agren12}, iMAT \cite{Zur10}, MBA \cite{Jerby10} and GIMME \cite{Becker08} are the most commonly used (see \tab{table3} for an overview). Here we provide a short outline of the different algorithms, and refer to \cite{Blazier12} for a more extensive overview. For GIMME, iMAT, and MBA, we briefly discuss some notable differences to {\rm\textsc{fastcore}}.

GIMME \cite{Becker08} takes as input microarray data and a biological function to optimize for, such as biomass production. GIMME starts by removing reactions with associated expression levels below a user-defined threshold, and then it optimizes for the specified biological function using linear programming. In case the pruning steps compromise the input biological function, GIMME reintroduces some previously removed reactions that are in minimal disagreement with the expression data.
Since GIMME has not been designed to include all core reactions in the solution (as {\rm\textsc{fastcore}} does), the reconstructions obtained by GIMME and {\rm\textsc{fastcore}} can differ significantly: \nikos{Running the {\em createTissueSpecific} function of the COBRA toolbox on a set of liver core reactions (see Section 3) treating them as expressed reactions (and adding a biomass reaction \cite{Wang12} and a sink reaction for glycogen to be used as optimization function), only about 50\% of the core reactions of the GIMME model were consistent at the solution. A fairer comparison would require adapting {\rm\textsc{fastcore}} to explicitly deal with omics data, which is outside the scope of the current work.}

iMAT \cite{Zur10} was originally designed for the integration of transcriptomic data. iMAT optimizes for the consistency between the experimental data and the activity state of the model reactions. iMAT tries to include modes composed of reactions associated to genes with high expression value, and therefore a threshold needs to be chosen to segregate between low, medium, and highly expressed genes. The computational demands of iMAT are high due to the repeated use of mixed integer linear programming. 
As with GIMME, direct comparison of iMAT to {\rm\textsc{fastcore}} is problematic. \nikos{Nevertheless, we applied iMAT (own implementation) on the liver problem (see Section 3), by setting the liver core reactions to RH (reaction high) and all non-core reactions to RL (reaction low). iMAT determined 549 core reactions as active, while 182 and 338 reactions were classified as undetermined and inactive, respectively. This means that about 50\% of the core reactions were lost during iMAT model building. As with GIMME, this demonstrates the difficulty of directly comparing {\rm\textsc{fastcore}} to algorithms that optimize different objectives.}

mCADRE \cite{Wang12} is similar to MBA, except that the pruning order is not random, but it depends on the tissue-specific expression evidence and weighted connectivity to other reactions of the network. Reactions that are associated to genes that are never tagged as expressed and which are not connected to reactions associated to highly expressed genes are first evaluated in the pruning step. Reactions are effectively removed if the removal does not impair core reactions and metabolic functions to carry a flux (mCADRE removes core reactions if the core/non-core reaction ratio is below a user-given threshold). mCADRE uses mixed integer linear programming and therefore it does not scale up to large networks (but it is in general faster than MBA). 

INIT \cite{Agren12} uses data retrieved from public databases in order to assess the presence of a certain reaction-respective metabolites in the cell type of interest. INIT uses mixed integer linear programming to build a model in which all reactions can carry a flux. Contrary to other algorithms, INIT does not rely on the assumption of a steady state, but it allows small net accumulation of all metabolites of the model.}

The closest algorithm to {\rm\textsc{fastcore}} is the MBA algorithm of Jerby et al.~\cite{Jerby10}. MBA takes as input two core sets of reactions, and it searches for a consistent network that contains all reactions from the first set, a maximum number of reactions from the second set (for a given tradeoff), and a minimal number of reactions from the global network. ({\rm\textsc{fastcore}} can be easily adapted to work with multiple core sets, by introducing a set of weights that reflect the confidence of each reaction to be active in the given context, and adding appropriate regularization terms in the objective functions of \lp{cardinality} and \lp{l1norm} that capture the given tradeoff. We will address this variant in future work.) Both {\rm\textsc{fastcore}} and MBA involve a search for a minimal consistent subnetwork, however the search strategy of {\rm\textsc{fastcore}} is very different to MBA: Whereas {\rm\textsc{fastcore}} iteratively expands the active set $\A$ starting with $\A=\emptyset$, MBA starts with $\A = \N$ and iteratively prunes the set $\A$ by checking whether the removal of each individual reaction (selected in random order) compromises network consistency. \nikos{As the pruning order affects the output model, this step of MBA is repeated multiple times. MBA builds a final model by adding one by one non-core reactions with the highest presence rate over all pruning runs, and it stops when a consistent final model is obtained. Due to the multiple pruning runs, MBA has very high computational demands.}
Consistency testing in MBA is carried out with the CMC algorithm that is based on \lp{jerby}, as explained earlier. 
Hence, {\rm\textsc{fastcore}}'s search strategy differs to MBA in two key aspects: First, consistency testing in {\rm\textsc{fastcore}} involves the maximization of flux cardinality (\lp{cardinality}) instead of sum of fluxes (\lp{jerby}), which results in fewer LP iterations. Second, the search for compact solutions in {\rm\textsc{fastcore}} involves $L_1$-norm minimization instead of pruning. The advantage of the former is that it can be encoded by a single LP, resulting in significant overall speedups (see Section 3.2).

\section{Results}

\begin{figure}
	\centering
	\newlength\figureheight 
	\newlength\figurewidth 
	\setlength\figureheight{4cm} 
	\setlength\figurewidth{4cm}
\includegraphics[scale=0.7]{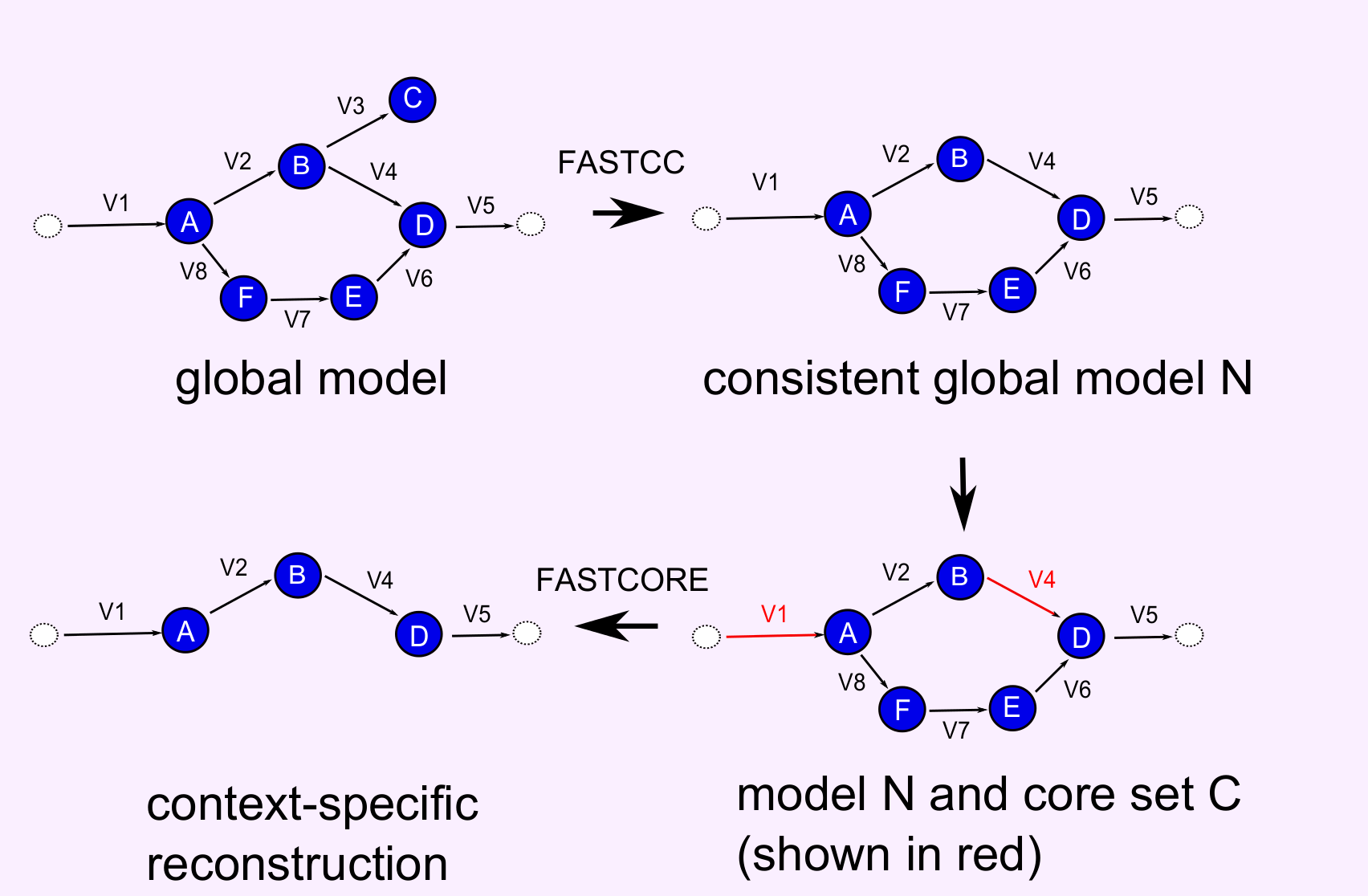}
	\caption{\nikos{Flowchart of the overall pipeline for generating consistent context-specific models.}}
	\label{fig:Flowchart}
\end{figure}

\nikos{Generic metabolic reconstructions like Recon\,2 are inconsistent models as they contain reactions that are not able to carry nonzero flux due to gaps in the network (see next section). The first step towards obtaining a consistent context-specific reconstruction is therefore to extract the consistent part of a global generic model. This can be achieved by {\rm \textsc{fastcc}} or other similar methods (see Section 2.2). The consistent global model serves then as input for the context-specific reconstruction with {\rm\textsc{fastcore}}. In \fig{Flowchart} we show a flowchart of the overall pipeline. }

We report results on two sets of problems, the first involving consistency verification of an input model, and the second involving the reconstruction of a context-specific model from an input model and a core set of reactions.
The {\rm\textsc{fastcore}} algorithm was implemented in the COBRA toolbox \cite{COBRA2}, using Matlab 2013a and the IBM CPLEX solver (version 12.5.0.0). Test runs were performed on a standard 1.8 GHz Intel Core i7 laptop with 4 GB RAM running Mac OS X 10.7.5. In all experiments we used flux threshold $\varepsilon=$1e-4.
The software is available from\,
\verb!bio.uni.lu/systems_biology/software/!

\subsection{Consistency testing}

In the first set of experiments we applied {\rm \textsc{fastcc}}, the consistency testing variant of {\rm\textsc{fastcore}}, for consistency verification of \nikos{four} input models, 
and compared it against the FastFVA algorithm of Gudmundsson and Thiele \cite{Gudmundsson10}, and an own implementation (based on {\rm \textsc{fastcc}} but with \lp{jerby} replacing \lp{cardinality}) of the CMC algorithm of Jerby et al.~\cite{Jerby10}. We also tested the FVA algorithm of the {\em reduceModel} function of the COBRA toolbox \cite{COBRA2}, and the MIRAGE algorithm of Vitkin and Shlomi \cite{Vitkin12}, but we do not include them in the results as they performed worse than the reported ones.
The input models were the following:
\begin{itemize}



\item \nikos{c-Yeast ($\#\N=1204$), the consistent part of a yeast model \cite{Zomorrodi10}.}

\item \nikos{c-Ecoli ($\#\N=1718$), the consistent part of an {\em E.~coli} model \cite{Orth11}. (Here we set to 1000 the upper bounds of all fluxes that were fixed to zero, and we multiplied all bounds by 1000 to avoid numerical issues.)}

\item c-Recon1 ($\#\N=2469$), the consistent part of Recon\,1 \cite{Duarte07}. (Recon\,1 was found to contain 1273 \nikos{blocked} reactions.)

\item c-Recon2 ($\#\N=5834$), the consistent part of Recon\,2 \cite{Thiele13}. (Recon\,2 was found to contain 1606 \nikos{blocked} reactions.)
\end{itemize}

\begin{table*}
\caption{Comparing {\rm \textsc{fastcc}} to fastFVA \cite{Gudmundsson10} and CMC \cite{Jerby10} on \nikos{four} input models.  \label{tab-table1}}
\begin{center}
\begin{tabular*}{\hsize}{@{\extracolsep{\fill}}lcrrccrrccrrccrrc}
~ &
 \multicolumn{4}{c}{\nikos{c-Yeast}} &
 \multicolumn{4}{c}{\nikos{c-Ecoli}}&
 \multicolumn{4}{c}{c-Recon1}&
 \multicolumn{4}{c}{c-Recon2} \\
 \hline
~&~& \#\,LPs & time$^*$\hspace*{-4pt} &~&~& \#\,LPs & time &~&~& \#\,LPs & time &~&~& \#\,LPs & time &~\\
\hline
fastFVA 	&~& 2408 &   3  &~&~& 3436 & 3	 &~&~& 4938 & 9 	 &~&~& 11668 & 207	&\\ 
CMC     &~&   18 & 0.5  &~&~&   25 & 1	 &~&~&   49 & 2   	 &~&~&    42 &  11 &\\ 
{\rm \textsc{fastcc}} &~&    7 & 0.1  &~&~&    2 & 0.2  &~&~&    9 & 0.4  &~&~&    19 &   5&\\ 
\hline
\multicolumn{16}{l}{$^*$\footnotesize in seconds}
\end{tabular*}
\end{center}
\end{table*}

The results are shown in \tab{table1}. {\rm \textsc{fastcc}} is faster and it uses much fewer LPs than the other two algorithms. We note that fastFVA is based on an optimized Matlab/C++ implementation with LP warm-starts, while {\rm \textsc{fastcc}} is based on standard Matlab.
These results confirm the appropriateness of flux cardinality (\lp{cardinality}) as a metric for network consistency testing, in agreement with the theoretical analysis and the discussions above.

\subsection{Reconstruction of a liver model}

In the second set of experiments, we used the {\rm\textsc{fastcore}} algorithm to reconstruct a liver specific metabolic network model from the consistent part of Recon\,1 (c-Recon1, $\#\N=2469$), and we compared against an own implementation of the MBA algorithm of Jerby et al.~\cite{Jerby10}.
We applied the two algorithms in two settings. The first setting involves the liver specific input reaction set of Jerby et al.~\cite{Jerby10}, which is based on 779 `high' core and 290 `medium' core reactions (the latter set is supported by weaker biological evidence than the former). To allow a comparison with {\rm\textsc{fastcore}}, we defined a single core set as the union of the high and medium core reaction sets, and we applied the two algorithms on this core set. 
The second setting uses the `strict' liver model of Jerby et al.~\cite{Jerby10}, which contains 1083 high core reactions and no medium core reactions, and therefore allows a direct comparison with {\rm\textsc{fastcore}}.

\begin{table*}
\caption{Comparing {\rm\textsc{fastcore}} to MBA \cite{Jerby10} on liver model reconstruction from c-Recon1.\label{tab-table2}}
\begin{center}
\begin{tabular*}{\hsize}{@{\extracolsep{\fill}}lCPPPPCCPPPP}
~ &~& \multicolumn{4}{c}{liver core set ($\#\C=1069$)} &~&~&
 \multicolumn{4}{c}{strict liver core set ($\#\C=1083$)} \\
 \hline
~&~& $\#\A~$ & IR$^*$\hspace*{-6pt} $~$  & \#LPs & time$^\ddag$\hspace*{-4pt} &~&~& $\#\A~$ & IR$~$ & \#LPs & time \\
\hline
MBA &~&  1826 & 1573 & 72279 & 7383 &~&~ & 1888 & 1630 & 71546 & 6730 \\ 
{\rm\textsc{fastcore}} &~&  1746 & 1546 & 20 & 1 &~& ~ & 1818 & 1627 & 20 & 1 \\ 
\hline
\multicolumn{12}{l}{$^*$\footnotesize number of intracellular reactions} \\
\multicolumn{12}{l}{\parbox{\linewidth}{$^\ddag$\footnotesize 
the reported time (in seconds), as well as the number of LPs, refer to a single pruning step of MBA, \nikos{whereas $\#\A$ and IR refer to the full MBA}.} \hspace*{-95pt}} 
\end{tabular*}
\end{center}
\end{table*}

The results for the two settings are shown in \tab{table2}. \nikos{We note that for MBA, the reported number of LPs and the runtime refer to a single pruning iteration of the algorithm, whereas the size of each reconstruction refers to the final model after 1000 pruning iterations.} In both settings, {\rm\textsc{fastcore}} is several orders of magnitude faster than MBA, achieving a full reconstruction of a liver specific model in about one second, using a much smaller number of LPs. 
As MBA employs a greedy pruning strategy for optimization, the number of LPs that it uses and its total runtime can be very high, as also indicated by Wang et al.~\cite{Wang12} who reported runtime of a single pruning pass of MBA in the order of 10 hours on a 2.34 GHz CPU computer.

The reconstructed models by {\rm\textsc{fastcore}} are also more compact than those obtained by MBA, with a difference of 70-80 non-core reactions. \nikos{For the standard liver model, 1687 out of the 1746 reactions (96\%) of the {\rm\textsc{fastcore}} reconstruction appear also in the MBA reconstruction, whereas for the strict liver model the common reactions are 1739 out of 1818 (95\%). 
The two algorithms turned out to use alternative transporters to connect the core reactions: In the standard liver model, 46 out of 59 reactions that are present exclusively in the {\rm\textsc{fastcore}} reconstruction are transporter reactions or other reactions which are not associated to a specific gene and thus are not sufficiently supported in the core set, whereas in MBA the corresponding numbers are 116 out of 139 reactions. 
Note that both MBA and {\rm\textsc{fastcore}} try to minimize the number of added non-core reactions in order to obtain a compact consistent model. The above difference in the number of added non-core reactions between MBA and {\rm\textsc{fastcore}} is the result of the different optimization approaches taken by the two algorithms, and no biological relevance should be attributed to each reconstruction other than the one implied by the makeup of the core set. From this point of view, {\rm\textsc{fastcore}} performs in general better than MBA, as it tends to add fewer unnecessary reactions.
}

\begin{figure}
	\centering
\includegraphics[scale=0.5]{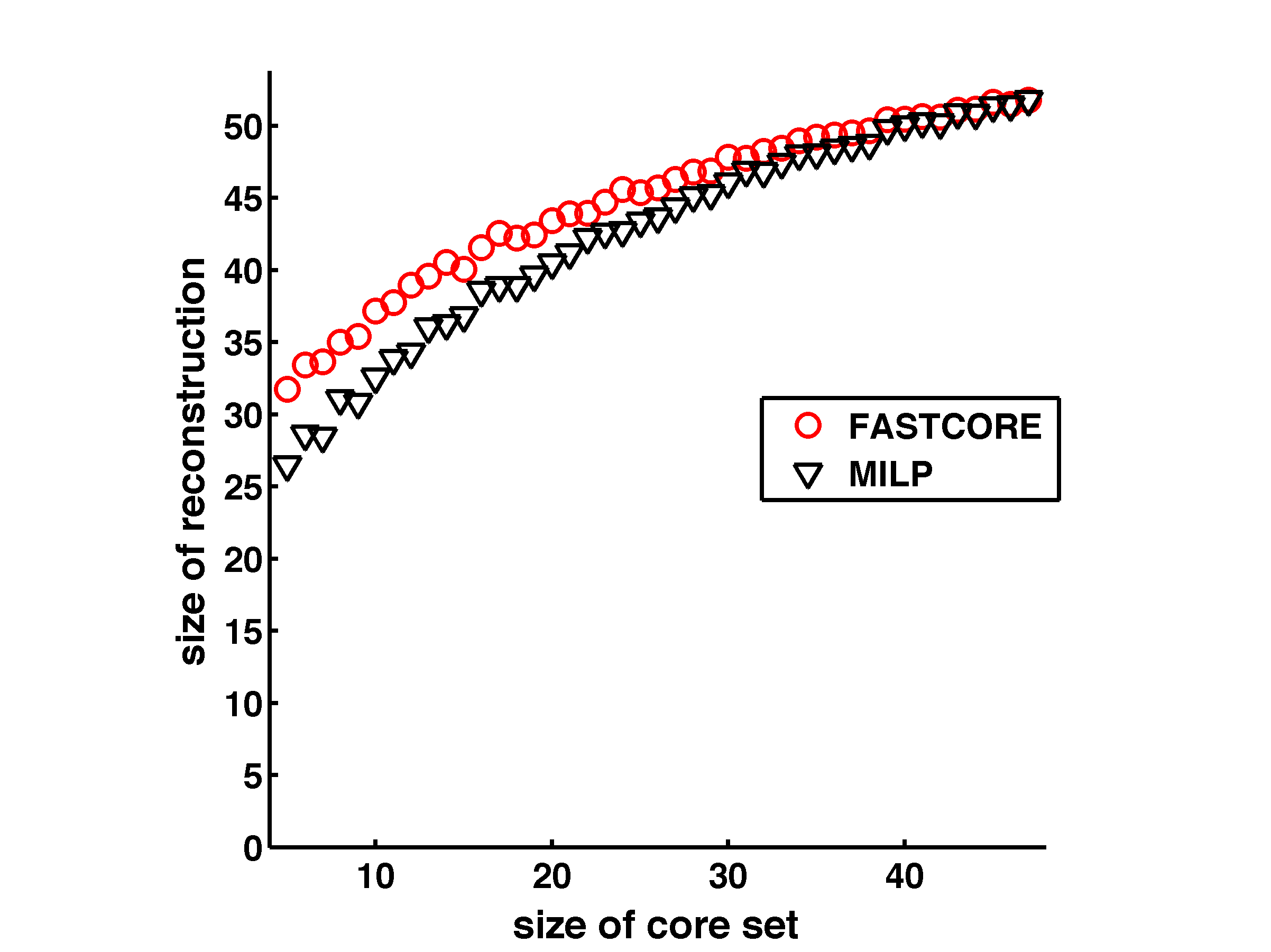}
	\caption{\nikos{Comparing {\rm\textsc{fastcore}} to an exact MILP solver on a small {\em E.~coli} model \cite{Orth10}. Shown are mean values of sizes of reconstructed models (over 50 repetitions for each core set; standard deviations were small and are omitted to avoid clutter) as a function of the size of the core set. {\rm\textsc{fastcore}} computes near-optimal reconstructions, which improve with the size of the core set.}}
	\label{fig:MILPvsFastcore}
\end{figure}

\nikos{We also compared the solutions of {\rm\textsc{fastcore}} to those of \milp{1}, using core sets that are randomly generated from a consistent subset of {\em E.~coli core} \cite{Orth10}. This is a small model with $\#\N=53$ and 414 elementary modes (unfortunately, the dependence of the \milp{1} solver to the number of elementary modes did not allow testing larger models). In \fig{MILPvsFastcore} we show the size of the reconstructed models (mean values) obtained with the exact MILP solver vs.\ {\rm\textsc{fastcore}}, as a function of the size of the core set. {\rm\textsc{fastcore}} is capable of obtaining very good approximations to the optimal solutions, which improve with the size of the core set. }

To evaluate {\rm\textsc{fastcore}}'s performance in correctly identifying liver reactions, we performed \nikos{repeated random sub-sampling validation in which {\rm\textsc{fastcore}} was used to reconstruct the liver metabolism based on a reduced, randomly selected `subcore' set of 80\% of the original core reactions. 
As in \cite{Jerby10}, we wanted to test whether {\rm\textsc{fastcore}} is able to recover a significant number of the 20\% left-out core reactions. To test for the enrichment of the left-out core reactions in the reconstructed model, we used a hypergeometric test, in which the total population is defined by all non-subcore reactions in the global network, the number of draws is defined as the number of non-subcore reactions included in the reconstruction, and the left-out core reactions are the `successes'. Under the null-hypothesis that there is no enrichment for the left-out core reactions when reconstructing the liver model based on the subcore set, we can compute a p-value for including at least the number of observed left-out core reactions in the reconstruction. We repeated this random sub-sampling procedure 500 times and computed the corresponding p-values. The median of these p-values was $0.0025$, indicating the ability of {\rm\textsc{fastcore}} to capture liver-specific reactions that were included in the original core set.}

As argued above, the reconstructions obtained by {\rm\textsc{fastcore}} need not optimize for cellular functions other than the ones implied by the composition of the input core set, and it is an interesting research question how to modify {\rm\textsc{fastcore}} so that it can explicitly capture functional requirements in its reconstructions. Nevertheless, it is of interest to test whether the current version of {\rm\textsc{fastcore}} can produce reconstructions that {\em are} functionally relevant, perhaps for slight variations of the core set. To this end, as in \cite{Jerby10}, we checked whether the (standard) liver model reconstructed by {\rm\textsc{fastcore}} can perform gluconeogenesis from glucogenic amino acids, glycerol, and lactate (altogether 21 metabolites). If not yet included, transporters from the extracellular medium to the cytosol were added to the model (glycerol, glutamate, glycine, glutamine, and serine). This was necessary as the transport reactions were not sufficiently supported in the core set. This `extended' liver model was able to convert 17/21 metabolites (vs 12/21 metabolites of the non-extended model). 
The extended liver model was then used to simulate the liver disorders hyperammonemia and hyperglutamenia, which affect the capacity to metabolize dietary amino acids into urea \cite{Jerby10}. Loss of function mutations of three enzyme-coding genes, argininosuccinate synthetase (ASS), argininosuccinate lyase (ASL), and ornithine transcarbamylase (OTC) were identified in patients suffering from these disorders. The rates of the reactions controlled by the three genes were fixed to 500, 250, or zero, to mimic the healthy homozygote (no mutation), heterozygote (loss of one allele), and the complete loss of function, respectively. 
To allow for a comparison with the experimental study of Lee et al.~\cite{Lee00} where labeled 15N-glutamine was administrated to patients suffering from inborn errors affecting the three genes, we explicitly shut down the influx of other potential nitrogen sources in the liver model, thereby simulating only the uptake and metabolism of glutamine. By allowing the influx of only one nitrogen source, the fate of the latter could be determined exactly in the model.
The ratio of urea secretion level over glutamine absorption was computed by sampling over the feasible space \cite{Price04}. In accordance with the wet lab observations \cite{Lee00}, the severity of the disorders, characterized by the mean urea over glutamine ratio, increased with the level of loss of function of the three genes ASS, ASL, and OTC (see \fig{Urea}). 
Null patients showed no native production of urea. Overall, the ratios predicted by the {\rm\textsc{fastcore}} model faithfully match the experimentally observed ones \cite{Lee00}. (The corresponding ratios reported by Jerby et al.\ when using the MBA algorithm~\cite{Jerby10} matched less well the experimental observations, probably because of the cross-feeding of nitrogen to urea from multiple nitrogen sources. By running the above procedure on the MBA model, we noticed that both models attained comparable urea / glutamine flux ratios.) To summarize, the above experiments demonstrate that, by an informed choice of the core set and influx bounds, {\rm\textsc{fastcore}} can indeed give rise to functionally relevant models.

\begin{figure}
	\centering
	\includegraphics[scale=1]{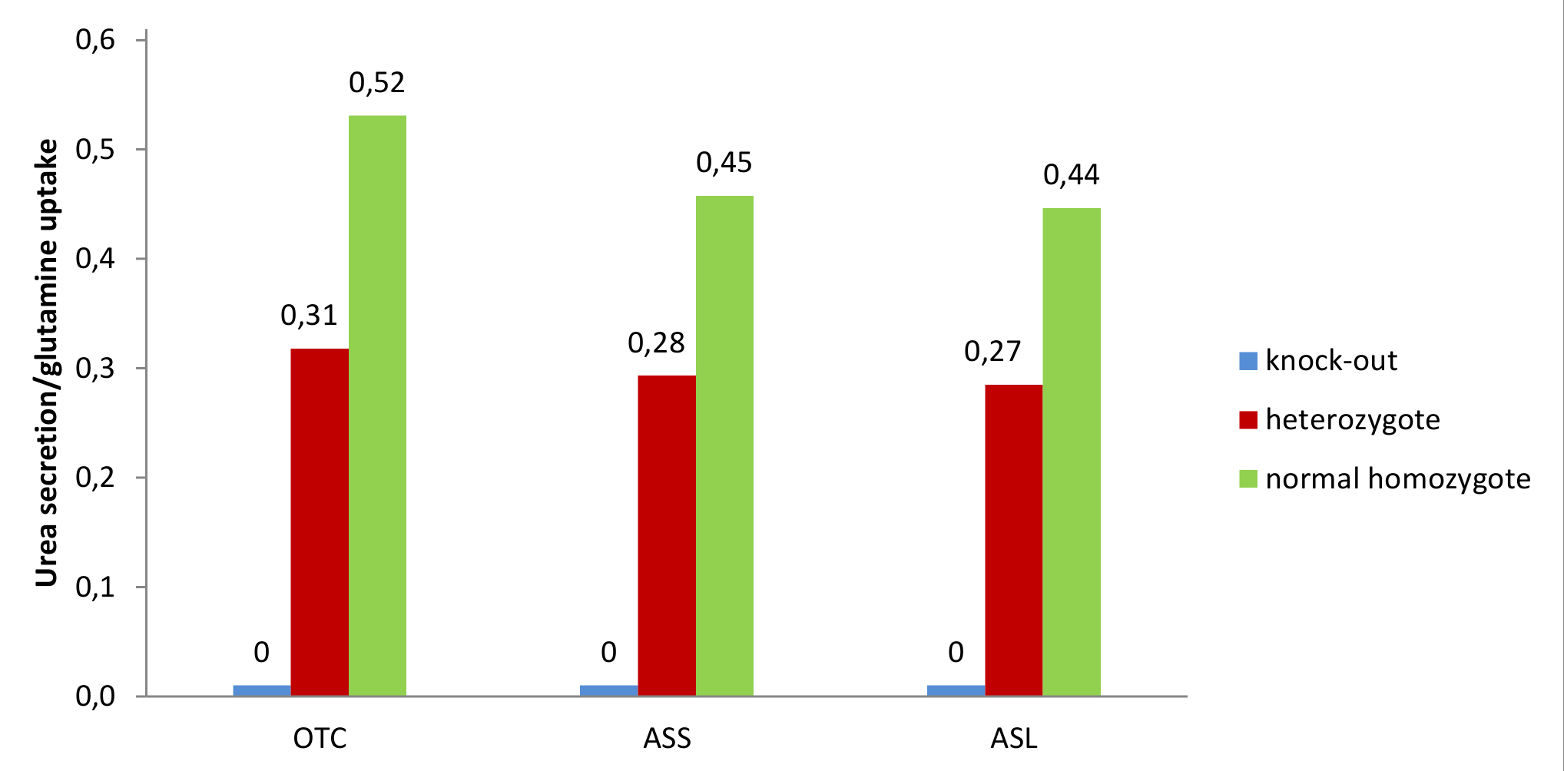}
	\caption{{Mean urea/glutamine ratio in the extended liver model obtained by {\rm\textsc{fastcore}}.} Healthy (normal homozygote), partial (heterozygote) and full knock-out cases. See text for details. }
	\label{fig:Urea}
\end{figure}

\subsection{Reconstruction of a murine macrophage model}

We also used the {\rm\textsc{fastcore}} algorithm to build a cell-type specific murine macrophage model from the consistent part of Recon1bio (comprising $\#\N=2474$ reactions). Recon1bio ($\#\N=3745$) is a modified Recon\,1 model that contains three extra reactions (biomass, NADPOX, and a sink reaction to balance the glycogenin self-glucosylation reaction) \cite{Bordbar12}. We used a core set comprising 300 (out of 382) proteomics derived Raw264.7 macrophage reactions, as described by Bordbar et al.~\cite{Bordbar12}. (The remaining 82 reactions could not be added to the core set as they are situated in an inconsistent region of Recon\,1 and therefore carry a permanent zero net flux.)
For their macrophage reconstruction, Bordbar et al.\ used, among other methods, GIMMEp---a variant of the GIMME algorithm \cite{Becker08} that is similar to the MBA algorithm---and they obtained a network model containing 1026 intracellular reactions. Our main interest was to investigate whether {\rm\textsc{fastcore}} can obtain a functional network that is at least as compact as the one obtained with GIMMEp. {\rm\textsc{fastcore}} generated (in about one second and using 11 LPs) a consistent network model of 953 reactions, 831 of which are intracellular reactions. This is a much more compact model than the one obtained with GIMMEp.

\section{Discussion}

{\rm\textsc{fastcore}} is a generic algorithm for context-specific metabolic network reconstruction from genome-wide metabolic models, and it was motivated by requirements of fast computation and compactness of the output model. 

The key advantage of having a fast reconstruction algorithm is that it permits the execution of multiple runs in order to optimize for extra parameters or test different core sets extracted from the input data \cite{Folger11,Wang12}. For example, when working with gene expression data, the definition of the core set may depend on the threshold used to segregate between high expression genes (core reactions) and low expression genes (non-core reactions) \cite{Becker08}. As the choice of threshold is rather arbitrary, a practical approach could involve evaluating the robustness of the output model as a function of the chosen threshold. {\rm\textsc{fastcore}} can perform this analysis in a few minutes, whereas for the same problem other algorithms would need hours or days. 
(Algorithms like GIMME or GIMMEp that require manual curation and assembly of subnetworks, would also fail in this kind of task.) Another example where fast computation is imperative is cross-validation. In the current study (see Section 3) we ran a \nikos{random sub-sampling validation} procedure 500 times, an operation that took a few minutes with {\rm\textsc{fastcore}} but that would barely be manageable with other reconstruction algorithms. Other examples where fast computation is important are time-course experiments or experiments involving different patients or conditions \cite{Jerby12}. There {\rm\textsc{fastcore}} could more easily identify differential models over time and/or input conditions. 

Compactness is a key concept in various research areas of biology, such as the minimal genome \cite{Morowitz84,Maniloff96}. 
Notwithstanding, the requirement of model compactness seems to be in disagreement with the observation that biological systems are fairly redundant and this redundancy serves a specific purpose, namely, the fast adaptation to changes in the environment. Alternative pathways that perform similar functions are known to be expressed in different environmental conditions, allowing for instance to metabolize another type of sugar when glucose is not available \cite{Suckow96}. 
At any rate, the pursuit of compactness in metabolic network reconstruction need not be in conflict with the notion of redundancy. Alternative pathways will be included in a reconstructed model as long as `redundant' reactions that are supported by biological evidence are included in the core set.

\section*{Acknowledgments}
We would like to thank Ines Thiele, Ronan Fleming, Nils Christian, Evangelos Symeonidis, Nathan Price, and Rudi Balling for their feedback.

\small
\bibliographystyle{plain}

\begin{thebibliography}{10}

\bibitem{Acuna09}
V.~Acu{\~n}a, F.~Chierichetti, V.~Lacroix, A.~Marchetti-Spaccamela, M.~F.
  Sagot, and L.~Stougie.
\newblock Modes and cuts in metabolic networks: Complexity and algorithms.
\newblock {\em Biosystems}, 95(1):51--60, 2009.

\bibitem{Agren12}
R.~Agren, S.~Bordel, A.~Mardinoglu, N.~Pornputtapong, I.~Nookaew, and
  J.~Nielsen.
\newblock Reconstruction of genome-scale active metabolic networks for 69 human
  cell types and 16 cancer types using {INIT}.
\newblock {\em PLoS Computational Biology}, 8(5):e1002518, 2012.

\bibitem{Becker08}
S.~A. Becker and B.~{\O}. Palsson.
\newblock Context-specific metabolic networks are consistent with experiments.
\newblock {\em PLoS computational biology}, 4(5):e1000082, 2008.

\bibitem{Blazier12}
A.~S. Blazier and J.~A. Papin.
\newblock Integration of expression data in genome-scale metabolic network
  reconstructions.
\newblock {\em Frontiers in Physiology}, 3:299, 2012.

\bibitem{Bordbar12}
A.~Bordbar, M.~L. Mo, E.~S. Nakayasu, A.~C. Schrimpe-Rutledge, Y.~M. Kim, T.~O.
  Metz, M.~B. Jones, B.~C. Frank, R.~D. Smith, S.~N. Peterson, et~al.
\newblock Model-driven multi-omic data analysis elucidates metabolic
  immunomodulators of macrophage activation.
\newblock {\em Molecular Systems Biology}, 8(1), 2012.

\bibitem{Boyd04}
S.~Boyd and L.~Vandenberghe.
\newblock {\em Convex Optimization}.
\newblock Cambridge University Press, Cambridge, UK, 2004.

\bibitem{Price10}
S.~Chandrasekaran and N.~D. Price.
\newblock Probabilistic integrative modeling of genome-scale metabolic and
  regulatory networks in {E}scherichia coli and {M}ycobacterium tuberculosis.
\newblock {\em Proceedings of the National Academy of Sciences},
  107(41):17845--17850, 2010.

\bibitem{Chang10}
R.~L. Chang, L.~Xie, L.~Xie, P.~E. Bourne, and B.~{\O}. Palsson.
\newblock Drug off-target effects predicted using structural analysis in the
  context of a metabolic network model.
\newblock {\em PLoS Computational Biology}, 6(9):e1000938, 2010.

\bibitem{Christian09}
N.~Christian, P.~May, S.~Kempa, T.~Handorf, and O.~Ebenh{\"o}h.
\newblock An integrative approach towards completing genome-scale metabolic
  networks.
\newblock {\em Mol. BioSyst.}, 5(12):1889--1903, 2009.

\bibitem{Chvatal79}
V.~Chv{\'a}tal.
\newblock A greedy heuristic for the set-covering problem.
\newblock {\em Mathematics of operations research}, 4(3):233--235, 1979.

\bibitem{Debardinis12}
R.~J. DeBerardinis and C.~B. Thompson.
\newblock Cellular metabolism and disease: what do metabolic outliers teach us?
\newblock {\em Cell}, 148(6):1132--1144, 2012.

\bibitem{Dreyfuss13}
J.~M. Dreyfuss, J.~D. Zucker, H.~M. Hood, L.~R. Ocasio, M.~S. Sachs, and J.~E.
  Galagan.
\newblock Reconstruction and validation of a genome-scale metabolic model for
  the filamentous fungus {N}eurospora crassa using {FARM}.
\newblock {\em PLoS Comput Biol}, 9(7):e1003126, 07 2013.

\bibitem{Duarte07}
N.~C. Duarte, S.~A. Becker, N.~Jamshidi, I.~Thiele, M.~L. Mo, T.~D. Vo,
  R.~Srivas, and B.~{\O}. Palsson.
\newblock Global reconstruction of the human metabolic network based on genomic
  and bibliomic data.
\newblock {\em Proceedings of the National Academy of Sciences},
  104(6):1777--1782, 2007.

\bibitem{Fleming12}
R.~M.~T. Fleming, C.~Maes, Y.~Ye, M.~A. Saunders, and B.~{\O}. Palsson.
\newblock A variational principle for computing nonequilibrium fluxes and
  potentials in genome-scale biochemical networks.
\newblock {\em Journal of Theoretical Biology}, 292:71--77, 2012.

\bibitem{Folger11}
O.~Folger, L.~Jerby, C.~Frezza, E.~Gottlieb, E.~Ruppin, and T.~Shlomi.
\newblock Predicting selective drug targets in cancer through metabolic
  networks.
\newblock {\em Molecular systems biology}, 7:501, 2011.

\bibitem{Gagneur04}
J.~Gagneur and S.~Klamt.
\newblock Computation of elementary modes: a unifying framework and the new
  binary approach.
\newblock {\em BMC bioinformatics}, 5(1):175, 2004.

\bibitem{Gudmundsson10}
S.~Gudmundsson and I.~Thiele.
\newblock Computationally efficient flux variability analysis.
\newblock {\em BMC bioinformatics}, 11(1):489, 2010.

\bibitem{Hao10}
T.~Hao, H.~W. Ma, X.~M. Zhao, and I.~Goryanin.
\newblock Compartmentalization of the {E}dinburgh human metabolic network.
\newblock {\em BMC bioinformatics}, 11(1):393, 2010.

\bibitem{Hiller13}
K.~Hiller and C.~M. Metallo.
\newblock Profiling metabolic networks to study cancer metabolism.
\newblock {\em Current Opinion in Biotechnology}, 24:60--68, 2013.

\bibitem{Jensen11}
P.~A. Jensen and J.~A. Papin.
\newblock Functional integration of a metabolic network model and expression
  data without arbitrary thresholding.
\newblock {\em Bioinformatics}, 27(4):541--547, 2011.

\bibitem{Jerby12}
L.~Jerby and E.~Ruppin.
\newblock Predicting drug targets and biomarkers of cancer via genome-scale
  metabolic modeling.
\newblock {\em Clinical Cancer Research}, 18(20):5572--5584, 2012.

\bibitem{Jerby10}
L.~Jerby, T.~Shlomi, and E.~Ruppin.
\newblock Computational reconstruction of tissue-specific metabolic models:
  Application to human liver metabolism.
\newblock {\em Molecular Systems Biology}, 6:401, 2010.

\bibitem{Julius08}
A.~A. Julius, M.~Imielinski, and G.~J. Pappas.
\newblock Metabolic networks analysis using convex optimization.
\newblock In {\em 47th IEEE Conference on Decision and Control}, pages
  762--767, 2008.

\bibitem{Lee00}
Brendan Lee, Hong Yu, Farook Jahoor, William O'Brien, Arthur~L. Beaudet, and
  Peter Reeds.
\newblock In vivo urea cycle flux distinguishes and correlates with phenotypic
  severity in disorders of the urea cycle.
\newblock {\em Proceedings of the National Academy of Sciences},
  97(14):8021--8026, 2000.

\bibitem{Lewis10}
N.~E. Lewis, G.~Schramm, A.~Bordbar, J.~Schellenberger, M.~P. Andersen, J.~K.
  Cheng, N.~Patel, A.~Yee, R.~A. Lewis, R.~Eils, et~al.
\newblock Large-scale in silico modeling of metabolic interactions between cell
  types in the human brain.
\newblock {\em Nature biotechnology}, 28(12):1279--1285, 2010.

\bibitem{Mahadevan03}
R.~Mahadevan and C.~H. Schilling.
\newblock The effects of alternate optimal solutions in constraint-based
  genome-scale metabolic models.
\newblock {\em Metabolic engineering}, 5(4):264, 2003.

\bibitem{Maniloff96}
J.~Maniloff.
\newblock The minimal cell genome: ``on being the right size''.
\newblock {\em Proceedings of the National Academy of Sciences}, 93(19):10004,
  1996.

\bibitem{Morowitz84}
H.~J. Morowitz.
\newblock The completeness of molecular biology.
\newblock {\em Israel journal of medical sciences}, 20(9):750, 1984.

\bibitem{Orth11}
J.~D. Orth, T.~M. Conrad, J.~Na, J.~A. Lerman, H.~Nam, A.~M. Feist, and
  B.~{\O}. Palsson.
\newblock A comprehensive genome-scale reconstruction of {E}scherichia coli
  metabolism.
\newblock {\em Molecular systems biology}, 7(1), 2011.

\bibitem{Orth10}
J.~D. Orth, R.~M.~T. Fleming, and B.~{\O}. Palsson.
\newblock Reconstruction and use of microbial metabolic networks: the core
  \emph{Escherichia coli} metabolic model as an educational guide.
\newblock In A.~B\"ock, R.~Curtiss~III, J.~B. Kaper, P.~D. Karp, F.~C.
  Neidhardt, T.~Nyström, J.~M. Slauch, C.~L. Squires, and D.~Ussery, editors,
  {\em {\emph{Escherichia coli} and \emph{Salmonella}: Cellular and Molecular
  Biology}}. ASM Press, Washington, DC, 2010.

\bibitem{Pourfar13}
M.~Pourfar, M.~Niethammer, and D.~Eidelberg.
\newblock Metabolic networks in {P}arkinson's disease.
\newblock In G.~Grimaldi and M.~Manto, editors, {\em Mechanisms and Emerging
  Therapies in Tremor Disorders}, Contemporary Clinical Neuroscience, pages
  403--415. Springer New York, 2013.

\bibitem{Price04}
N.~D. Price, J.~L. Reed, and B.~{\O}. Palsson.
\newblock Genome-scale models of microbial cells: evaluating the consequences
  of constraints.
\newblock {\em Nature Reviews Microbiology}, 2(11):886--897, 2004.

\bibitem{COBRA2}
J.~Schellenberger, R.~Que, R.~M.~T. Fleming, I.~Thiele, J.~D. Orth, A.~M.
  Feist, D.~C. Zielinski, A.~Bordbar, N.~E. Lewis, S.~Rahmanian, J.~Kang, D.~R.
  Hyduke, and B.~{\O} Palsson.
\newblock Quantitative prediction of cellular metabolism with constraint-based
  models: the {COBRA} {T}oolbox v2.0.
\newblock {\em Nat Protoc}, 6(9):1290--1307, Sep 2011.

\bibitem{Schuster94}
S.~Schuster and C.~Hilgetag.
\newblock On elementary flux modes in biochemical reaction systems at steady
  state.
\newblock {\em Journal of Biological Systems}, 2(2):165--182, 1994.

\bibitem{Shlomi08}
T.~Shlomi, M.~N. Cabili, M.J. Herrg{\aa}rd, B.~{\O}. Palsson, and E.~Ruppin.
\newblock Network-based prediction of human tissue-specific metabolism.
\newblock {\em Nature biotechnology}, 26(9):1003--1010, 2008.

\bibitem{Stephanopoulos98}
G.~Stephanopoulos, A.~A. Aristidou, and J.~Nielsen.
\newblock {\em Metabolic engineering: principles and methodologies}.
\newblock Academic Press, 1998.

\bibitem{Suckow96}
J.~Suckow, P.~Markiewicz, L.~G. Kleina, J.~Miller, B.~Kisters-Woike,
  B.~M{\"u}ller-Hill, et~al.
\newblock Genetic studies of the {L}ac repressor. {XV}: 4000 single amino acid
  substitutions and analysis of the resulting phenotypes on the basis of the
  protein structure.
\newblock {\em Journal of molecular biology}, 261(4):509, 1996.

\bibitem{Thiele10}
I.~Thiele and B.~{\O}. Palsson.
\newblock A protocol for generating a high-quality genome-scale metabolic
  reconstruction.
\newblock {\em Nature protocols}, 5(1):93--121, 2010.

\bibitem{Thiele13}
I.~Thiele, N.~Swainston, R.~M.~T. Fleming, A.~Hoppe, S.~Sahoo, M.~K. Aurich,
  H.~Haraldsdottir, M.~L. Mo, O.~Rolfsson, M.~D. Stobbe, et~al.
\newblock A community-driven global reconstruction of human metabolism.
\newblock {\em Nature biotechnology (doi:10.1038/nbt.2488)}, 2013.
\newblock doi:10.1038/nbt.2488.

\bibitem{Vitkin12}
E.~Vitkin and T.~Shlomi.
\newblock {MIRAGE}: a functional genomics-based approach for metabolic network
  model reconstruction and its application to cyanobacteria networks.
\newblock {\em Genome biology}, 13(11):R111, 2012.

\bibitem{Wang12}
Y.~Wang, J.~A. Eddy, and N.~D. Price.
\newblock Reconstruction of genome-scale metabolic models for 126 human tissues
  using {mCADRE}.
\newblock {\em BMC Systems Biology}, 6(1):153, 2012.

\bibitem{Zomorrodi10}
A.~R. Zomorrodi and C.~D. Maranas.
\newblock Improving the i{MM904} {S}.cerevisiae metabolic model using
  essentiality and synthetic lethality data.
\newblock {\em BMC systems biology}, 4(1):178, 2010.

\bibitem{Zur10}
Hadas Zur, Eytan Ruppin, and Tomer Shlomi.
\newblock imat: an integrative metabolic analysis tool.
\newblock {\em Bioinformatics}, 26(24):3140--3142, 2010.

\end{thebibliography}

\end{document}